\documentclass[runningheads]{llncs}
\usepackage[T1]{fontenc}   
\usepackage{amsmath}
\usepackage{graphicx}
\newtheorem{fact}{Fact}
\newcommand{\junk}[1]{}

\begin{document}





\title{The Voronoi Diagram of Weakly Smooth Planar Point Sets
  in $O(\log n)$ Deterministic Rounds on the Congested Clique}
\author{ Jesper Jansson
  \inst{1}
  \and
  Christos Levcopoulos
  \inst{2}
  \and \\
  Andrzej Lingas
  \inst{2}
}
\institute{
Graduate School of Informatics, Kyoto University, Kyoto, Japan.
\email{jj@i.kyoto-u.ac.jp}
\and
Department of Computer Science, Lund University, Lund, Sweden. 
\email{$\{$Christos.Levcopoulos, Andrzej.Lingas$\}$@cs.lth.se}
}

\pagestyle{plain}
\maketitle

\begin{abstract}
  We study the problem of computing the Voronoi diagram of a set of
  $n^2$ points with $O(\log n)$-bit coordinates in the Euclidean plane
  in a substantially sublinear in $n$ number of rounds in the
  congested clique model with $n$ nodes. Recently,
  Jansson et al. have
  shown that if the points are uniformly at random distributed in a
  unit square then their Voronoi diagram within the square can be
  computed in $O(1)$ rounds with high probability (w.h.p.).  We show
  that if a very weak smoothness condition is satisfied by an
  input set of $n^2$ points with $O(\log n)$-bit coordinates in the
  unit square then the Voronoi diagram of the point set within the
  unit square can be computed in $O(\log n)$ rounds in this model.
  \junk{Next, we present a randomized approach to the construction of the
  intersection of $n^2$ half-spaces dual to points with $O(\log
  n)$-bit coordinates in the three dimensional Euclidean space. We can
  solve this intersection problem in $O(\log n)$ rounds on the
  congested clique with $n$ nodes w.h.p.  By the known reductions, we
  conclude that the convex hull of $n^2$ points with $O(\log n)$-bit
  coordinates in the three-dimensional space as well as the Voronoi
  diagram of $n^2$ points with $O(\log n)$-bit coordinates in the
  plane can be constructed in $O(\log n)$ rounds on the congested
  clique with $n$ nodes w.h.p.}
\end{abstract}


\section{Introduction}
The \emph{congested clique} is a relatively new model of
communication and computation introduced by Lotker \emph{et al.} in
2005 \cite{LP-SPP05}. It focuses on the cost of communication
between the nodes in a network, ignoring the cost of local computation
within each node. Hence, it can be seen as opposite to the Parallel
Random Access Machine (PRAM) model, studied extensively in the 80s and
90s. The PRAM model focuses on the computation cost and ignores the
communication cost \cite{ACG85}.

Originally, the complexity of dense graph problems has been studied in
the congested clique model under the following assumptions.  Each node
of the congested clique represents a distinct vertex of the input
graph and knows its neighborhood in the graph.
Every node also knows the unique ID numbers (between $1$ and $n$) of itself
and all the other nodes at the start of the computation.
The computation proceeds in rounds.  In each
round, each of the $n$ nodes can send a distinct message of $O(\log
n)$ bits to each other node and can perform unlimited local
computation.
The primary complexity
objective is to minimize the number of rounds necessary to solve a
given problem on the input graph in this model.

For several basic graph problems, e.g., the minimum spanning tree
problem, one has succeeded to design even $O(1)$-round protocols in the
congested clique model \cite{K_21,R}.  Observe that when the input
graph is of bounded degree and edge weights have $O(\log n)$-bit representation, each node can send
the ID numbers of all nodes in its neighborhood and
the weights of its incident edges,
e.g., to  the first node in $O(1)$ rounds. After that, the
first node can solve the whole problem locally.  However, such a
trivial solution would require $\Omega(n)$ rounds when the input
graph is dense.
\junk{
\begin{figure}
\begin{center}
\includegraphics[scale=0.7]{computation1}
\end{center}
\caption{An example of a congested clique network.}
\label{fig: computation1}
\end{figure}}

Matrix problems \cite{CK15}, sorting and routing
\cite{L}, and geometric problems \cite{JLLP24} have also been studied in the congested
clique model.  In all cases, the basic input items, i.e., matrix
entries or keys, or points in the plane, respectively, are assumed to
have $O(\log n)$-bit representations and  each node initially
has a batch of $n$ such items. {\em Note that the bound
on bit representation of an input item is a
natural consequence of the $O(\log n)$-bit  bound on the size
of a single message which makes input items of unbounded bit
representation imcompatible with the assumed model.}
As in the graph case, in every round,
each node can send a distinct $O(\log n)$-bit message to each other
node and perform unlimited local computation.  Significantly, it has
been shown that matrix multiplication can be performed in
a number of rounds substantially sublinear in $n$ \cite{CK15}
while sorting and routing can be implemented
in  $O(1)$ rounds (Theorems 4.5 and 3.7 in \cite{L}).

As for the geometric problems, Jansson et al. \cite{JLLP24}
recently provided
low polylogarithmic, deterministic upper bounds
on the number of rounds required to solve
several basic geometric problems
for a set of $n^2$ points in the plane  with $O(\log
n)$-bit coordinates
in the model of congested clique with $n$ nodes.
As for the construction of the Voronoi diagram and the dual Delaunay
triangulation of the point set (see Fig.~\ref{fig: first} for an
illustration and Section~2 for the formal definition), they have shown
an $O(1)$ upper bound on the number of required rounds under the
assumption that the points are drawn uniformly at random from a unit
square. On the other hand, already at the end of 90s,
Goodrich presented $O(1)$-round randomized protocols for the
construction of three-dimensional convex hull of a set of points
in three-dimensional Euclidean space in $O(1)$ communication
rounds in the so-called Bulk Synchronous Processing model (BSP) \cite{God97}.
His result also implies an $O(1)$-round bound on the randomized construction
of the Voronoi diagram and the dual Delaunay triangulation
of a planar point set in the BSP model.
By using the $O(1)$-round routing protocol of Lenzen \cite{L},
Goodrich's $O(1)$ bound on the number of rounds necessary 
for the construction of the Voronoi diagram and Delaunay triangulation
most likely can be carried over from the BSP model to ours.

\begin{figure}
\begin{center}
\includegraphics[scale=0.88]{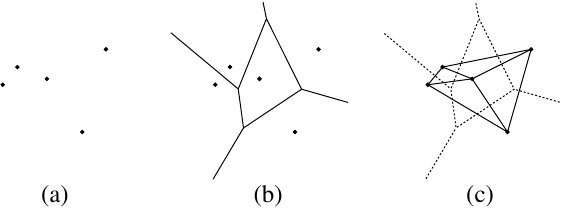}
\end{center}
\caption{An example of a planar point set, its Voronoi diagram, and the dual Delaunay triangulation.}
\label{fig: first}
\end{figure}

In this context, the major open problem is to derive a non-trivial upper bound
on the number of rounds sufficient to deterministically construct the Voronoi diagram
when the points are not necessarily drawn uniformly at random.
The bottleneck in the design of efficient parallel or
distributed algorithms for the Voronoi diagram of a planar point set
using a direct divide-and-conquer approach is 
an efficient parallel or distributed merging of Voronoi diagrams
Aggarwal et al. \cite{ACG85} presented an $O(\log^2 n)$-time CREW PRAM
algorithm for the Voronoi diagram
based on an involved $O(\log n))$-time PRAM method for the parallel merging.
Subsequently, Amato and Preparata \cite{AP95} demonstrated
an $O(\log n)$-time CREW PRAM algorithm for the three-dimensional
convex hull and consequently  also for the two-dimensional Voronoi
diagram of a point set.

We substantially extend the local approach to the
construction of the Voronoi diagram used in the design of parallel and
distributed algorithms for the Voronoi diagram of points drawn
uniformly at random, e.g., from a unit square \cite{JLLP24,LKL,VV}.
We show that already a very weak smoothness condition on the input set
of $n^2$ points with $O(\log n)$-bit coordinates within a unit square
is sufficient to obtain an $O(\log n)$ upper bound on the number of
rounds required to construct the Voronoi diagram of the set within the
unit square on the congested $n$-clique.  Roughly, our weak smoothness
condition says that if a square $Q$ of side length $\ell$ within the
unit square contains at least $n$ out of the $n^2$ input points then
any square of the same size at distance at most $4\sqrt 2\ell$ from $Q$ and
within the unit square has to contain at least one input point.

In order to simplify the presentation, we assume throughout the paper
that the points in the input point sets are in general position (i.e.,
neither any three input points are co-linear nor any four input points
are co-circular).

Our paper is structured as follows.  The next section contains basic
mathematical/geometric definitions, lemma, and facts on routing and
sorting in the congested clique model.  Section 3 presents our
protocol for the Voronoi diagram and Delaunay triangulation of a
weakly smooth planar point set within a square.  We conclude with
final remarks.

\section{Preliminaries}

The cardinality of a set $S$ is denoted  by $|S|.$

For a positive integer $r,$ $[r]$ stands for the set of
positive integers not exceeding $r.$

For a finite set $S$
of points in the Euclidean plane,
the \emph{Voronoi diagram of~$S$} is the partition of the plane into
$|S|$~regions such that each region consists of all points in the plane having
the same closest point in~$S$; see Fig. \ref{fig: first}.

A \emph{Delaunay triangulation of~$S$} is a maximal set of non-crossing edges
between pairs of points from~$S$ such that no point from~$S$ is placed inside any of
the formed triangles' circumcircles.
It is well known that if no four
points in $S$ are co-circular then
the Delaunay triangulation of~$S$ is a dual of the Voronoi
diagram of~$S$ in the following sense~\cite{PS}:
for each edge~$e$ of 
of each region in the Voronoi diagram
of $S$, if $e$ is a part of the bisector of the points
$u,\ v$ in $S$ then $(u,v)$ is an edge of the Delaunay
triangulation of $S$; again, see Fig. \ref{fig: first}.

Our concept of weak smoothness is formally defined in terms
of two parameters as follows.

\begin{definition}
Let $\varepsilon,\ d$ be two positive real constants.
A set of $N$ points in a unit square is {\em $(\varepsilon,d)$-smooth}
if for any two equal size  squares $Q,\ R$ within the unit square
the following implication holds:\\
if $Q$ contains at least $N^{\varepsilon}$ points of $S$
and $R$ is at distance at most $d\cdot \ell$ from $Q$,
where  $\ell$ is the length of each edge  of $Q$ and $R$,
then $R$ contains at least one point of $S$.
\end{definition}

We also need to define a sequence of grids within a unit square and related
notions.

\begin{definition}
For a nonnegative integer $i,$ we shall denote by $G_i(U)$ the orthogonal grid
within the unit orthogonal square  $U$ that includes the edges of $U$ such that the distance between
two neighboring vertical or horizontal grid line segments is $\frac 1 {2^i}.$
A {\em basic square} of $G_i(U)$ is a square within $U$ such that the endpoints
of each its edge is a pair of neighboring grid points.
For a basic square $R$ in $G_i(U),$ we shall denote the orthogonal
region consisting of $R$ and the two layers of basic squares around $R$ by $TL_i(R)$
(if between $R$ and an edge of the unit square there is place only for one
or zero layers then $TL_i(R)$ includes only one or zero layers on this side, respectively).
\end{definition}

\begin{figure}
\begin{center}
\includegraphics[scale=0.7]{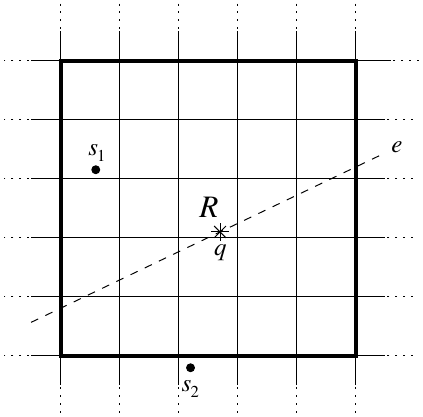}
\end{center}
\caption{An example of the configuration in  the proof of Lemma \ref{lem: first}.}
\label{fig: theo}
  \end{figure}

The proof of the following lemma corresponds to the second paragraph of
the proof of Theorem 4  in \cite{JLLP24}.

    \begin{lemma}\label{lem: first}
      Let $R$ be a basic square in a grid $G_i(U)$ within the unit square $U$.
      Consider a finite set $S$ of points within the unit square.
      If $R$ contains a point in $S$ then the Voronoi diagram of $S$
      within $R$ can be computed by taking into account only the points
      of $S$ within $TL_i(U).$
      Hence, in particular all edges $(u,v)$ of the Delaunay triangulation
      of $S$ such that a part of the bisector of $u$ and $v$ borders
      some region of the Voronoi diagram of $S$ within $R$ can be determined.
   \end{lemma}
   \begin{proof}
     Let $e$ be an edge of the Voronoi diagram of $S$ within $R$. The edge
         $e$ has to be a part of the bisector of some couple of points
         $s_1$ and $s_2$ in $S.$ Consider an arbitrary point $q$ on
         $e$.
     Suppose that $s_1$ or $s_2$ is placed outside $TL_i(R)$,
     i.e., the orthogonal area
     consisting of at
most $1+8+16=25$ squares including $R$.
See Fig. \ref{fig: theo}.  Without loss of
generality, let $s_2$ be such a point.  Then the distance between $q$ and $s_2$
is at least $2 \cdot \frac{1}{2^i}$, while the distance between $q$ and
every point inside~$R$ is at most $\sqrt{2} \cdot \frac{1}{2^i}$.  We
obtain a contradiction because $R$ contains at least one point from $S$
and $q$ is closer to such a point than to $s_2$.
\qed
   \end{proof}
\junk{
   The following fact is crucial for our randomized approach
   to the construction of Voronoi diagrams.

   \begin{fact} \label{fact: chernoff}(multiplicative Chernoff lower bound)
  Suppose $X_1, ..., X_n$ are independent random variables
  taking values in $\{0, 1\}.$
  Let X denote their sum and let $\mu= E[X]$
  denote the sum's expected value. Then, for any $\delta \in [0,1],$
  $Prob(X\le (1-\delta)\mu)\le e^{-\frac {\delta^2 \mu}2}$ holds.
  Similarly, for any $\delta \ge 0,$
  $Prob(X\ge (1+\delta)\mu)\le e^{-\frac {\delta^2 \mu}{2+\delta}}$ holds.  
  \end{fact}

We shall say that an event dependent on $n^2$ input
basic geometric objects in the two or three dimensional Euclidean space
     holds with high probability (w.h.p.) if its probability is
     at least $1-\frac 1 {n^{\alpha}}$ asymptotically (i.e., there is
     an integer $n_0$ such that for all $n\ge n_0,$ the probability is
     at least $1-\frac 1 {n^{\alpha}}$), where $\alpha $ is any
     constant not less than $2$.
}
Lenzen gave an efficient solution to the following fundamental routing problem
in the congested clique model, known as
     the {\em Information Distribution Task} (IDT) \cite{L}:\\
Each node of the congested $n$-clique holds a set of exactly $n$
$O(\log n)$-bit messages with their destinations, with multiple messages from
the same source node to the same destination node allowed.
Initially, the destination of each message is known only to its source node.
Each node is the
destination of exactly $n$ of the aforementioned messages. The
messages are globally lexicographically ordered by their source node, their
destination, and their number within the source node.  For simplicity,
each such message explicitly contains these values, in particular
making them distinguishable. The goal is to deliver all messages to their
destinations, minimizing the total number of rounds.

Lenzen proved that   IDT can be solved in $O(1)$ rounds (Theorem 3.7 in
\cite{L}). He also noted that the relaxed IDT, where each node is required
to send and receive at most $n$ messages, reduces to IDT
in $O(1)$ rounds.
From here on, we shall refer to this important result as:

\begin{fact}\cite{L}\label{fact: routing}
  The relaxed Information Distribution
  Task can be solved deterministically within $O(1)$ rounds.
  \end{fact}

The {\em Sorting Problem} (SP) is defined as follows:\\
Each node $i$ of the congested $n$-clique holds a set of $n$
$O(\log n)$-bit keys. All the keys are different w.l.o.g.  Each node
$i$ needs to learn all the keys of indices in $[n(i-1)+1,ni]$ (if any)
in the total order of all keys.

Lenzen showed that SP can be solved in $O(1)$ rounds if each node holds a
set of exactly $n$ keys (Theorem 4.5 in \cite{L}).  In order to relax
the requirement that each node holds exactly $n$ keys to that of with
most $n$ keys, we can determine the maximum key and add appropriate
different dummy keys in $O(1)$ rounds.
We summarize this result as:

\begin{fact}\cite{L}\label{fact: sorting}
    The relaxed Sorting Problem can be solved in $O(1)$ rounds.
  \end{fact}

\section{The local approach}
   Consider a $(\frac 12,4\sqrt 2)$-smooth set of $n^2$ points with $O(\log
   n)$-bit coordinates in a unit orthogonal square.  We shall first
   describe a protocol for listing the edges of the Delaunay
   triangulation of the set that are dual to the edges of the Voronoi
   diagram of the set within the unit square.  Roughly, it implicitly
   grows a quadtree of squares rooted at the unit square in phases
   corresponding to the levels of the quadtree.  If a square $R$
   currently at a leaf of the quadtree jointly with the two layers of
   equal size squares around it includes $O(n)$ input points then the
   intersection of the Voronoi diagram of the input point set with $R$
   and the dual edges of the Delaunay triangulation of the input point
   set can be computed locally. This follows from Lemma \ref{lem: first}
   combined with the fact that the parent square of $R$ does not
   satisfy an analogous condition.  Otherwise,
   four child squares whose union forms $R$ are created on the next
   level of the quadtree. In particular, checking the aforementioned
   condition in parallel for the squares at the current front level of
   the quadtree and delivering the necessary points to the nodes
   representing respective frontier squares in $O(1)$ rounds on the
   congested $n$-clique are highly non-trivial.
   \par
   \vskip 5pt
    \noindent
        {\bf protocol} $DT-SQUARE(S,U)$
        \par
        \noindent
            {\em Input}: A $(\frac 12,4\sqrt 2)$-smooth set of $n^2$ points with
            $O(\log n)$-bit coordinates in a unit orthogonal square $U$
            held in $n$-point batches at the $n$ nodes of the congested clique.
            \par
            \noindent
                {\em Output}:The set of the edges of the Delaunay triangulation of $S$
                dual to the edges of the Voronoi diagram of $S$ within $U$
                held in $O(n)$-edge batches at consecutive clique nodes.
                \begin{enumerate}
               \item
Initialize a list $L$ of edges of the Delaunay triangulation of $S.$
\item 
                Activate the basic square $U$ in $G_0(U)$ and assign it to the first node.
                \item
                For $i=0,1,\dots $ do
                \begin{enumerate}
                \item
                Each node for each point $p$ in its batch determines the number $num(p)$ of the basic square
                of $G_i(U)$ containing $p$ in a common fixed numbering of the
                basic squares in $G_i(U)$ (e.g., column-wise).
                Next, a {\em prefixed representation} of $p$ consisting of bit representation
                of $num(p)$ followed by the bit representation of the coordinates of $p$
                is created.
                \item
                  The points in $S$
                  are sorted by their prefixed representation.
                After that each node informs all other nodes about the range of numbers
                of the basic squares in $G_i(U)$ holding the prefixed representations of
                points in $S$ that landed at the node
                after the sorting of the prefixed representations of
                all the points.
              \item
                For each basic square $W$ in $G_i(U)$ such that
                the prefixed representations of points belonging
                to $W$ landed in a sequence $C$ of at least two
                consecutive nodes, the nodes in $C$ inform additionally the other nodes in $C$
                about the number of the prefixed representations of the points in $W$
                they got so in particular the node in $C$ with the smallest index can
                compute the total number of the points in $W$.
                \item
                Each node for each active basic square $R$ in $G_i(U)$ it represents sends
             queries to the nodes holding the prefixed  point representations
             of the points in the basic squares in
             $TL_i(R)$ (i.e., in $R$ and  the two layers of basic squares around $R$ in $G_i(U)$)
             about the number of points in these squares.
             In case several nodes hold the prefixed point representation of points
             in a basic square in $TL_i(U),$ the query is send just to that with the smallest
             index.
\item
After getting answers to the queries, each node for each active basic square in $R$ in $G_i(U)$ it represents
proceeds as follows.
If  the total  number of points of $S$ in $TL_i(R)$
                  does not exceed $100n$ then the node  asks the nodes holding the prefixed representations
                  of the points in the basic squares in $TL_i(R)$
                  for sending the points to the node.
                  After that the node computes the Voronoi diagram of all these points
                  and then the intersection of the diagram with $R$ locally.
                  Next, the node appends to $L$ all edges $(u,v)$ where a part of the bisector of $u$ and $v$
                  borders some region of the Voronoi diagram in the computed intersection.
Otherwise, the node activates the four squares in $G_{i+1}(U)$ whose union forms $R$
and assigns them temporarily to itself.
\item
The nodes balance the assignment of active basic squares in $G_{i+1}(U)$ by informing
all other nodes about the number of active basic squares in  $G_{i+1}(U)$  they are assigned
and following the results of the same assignment balancing algorithm run by each
of them separately locally.
\item
The list $L$ is sorted in order to remove multiple copies
of the same edge.
\end{enumerate}
\end{enumerate}

\begin{lemma}\label{lem: logn}
$DT-SQUARE(S,U)$ activates basic squares solely
in the grids $G_i(U)$, where $i=O(\log n).$
\end{lemma}
\begin{proof}
Simply, the points in $S$ have $O(\log n)$-bit coordinates so at depth at most $O(\log n)$ the
                condition in Step 3(e) of $DT-SQUARE(S,U)$ has to be
                satisfied.
 \qed
 \end{proof}

\begin{figure}
\begin{center}
\includegraphics[scale=0.5]{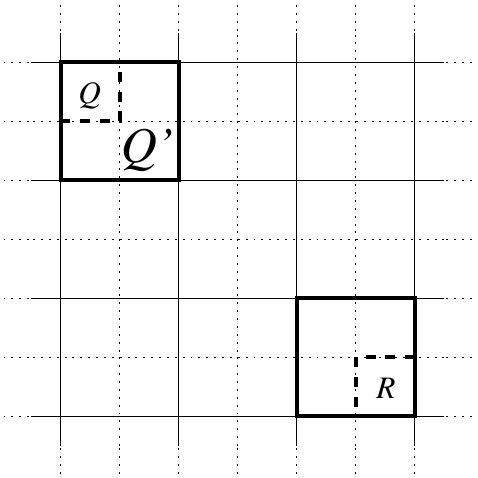}
\end{center}
\caption{An example of the configuration in  the proof of Lemma \ref{lem: correct}.}
\label{fig: correct}
  \end{figure}

\begin{lemma}\label{lem: correct}
The protocol $DT-SQUARE(S,U)$ is correct.
\end{lemma}
\begin{proof}
  When the Voronoi diagram of the points of $S$ in $TL_i(U)$
  for  a basic square $R$
  of the grid $G_i(U)$
is computed then there must be square $Q'$ in the grid $G_{i-1}(U)$
that contains at least $100n/25$ points in $S$
and is at distance at most $\frac {\sqrt 2} {2^{i-1}}$
from the basic square in $G_{i-1}(U)$ that is the parent of $R.$
Hence, there is a basic square $Q$ in $G_i(U)$ that is part of
$Q'$ and contains at least $100n/100$ points in $S$; see Fig. \ref{fig: correct}.
By straightforward calculations,
the distance between $Q$ and $R$ is at most $4\sqrt 2\frac 1 {2^i}.$
Thus, by the assumed $(\frac 1 {2}, 4\sqrt 2)$-smoothness property,
the square $R$ contains at least one point in $S.$
It follows from Lemma \ref{lem: first} that the intersection of the Voronoi diagram
of the points of $S$ in $TL_i(R)$
with $R$ yields the Voronoi diagram of
$S$ within $R.$ Hence, the edges appended to the list $L$ are
the edges of the Delaunay triangulation of $S$ dual
to the edges of the Voronoi diagram of $S$ within $U$.
It easily follows by induction on $i$ during forming
the quadtree of active basic squares  that the leaf active basic squares form
a partition of the unit square $U.$
Therefore, for each edge $(u,v)$
of the Delaunay triangulation of $S$ dual to an edge of the Voronoi diagram
of $S$ within $U$ there must exist
a positive integer $i$ and an active
basic square $R$ in $G_i(U)$ such that $R$ does not have any child
active basic squares in $G_{i+1}(U)$ and a part of the bisector
of $u$ and $v$ borders some region in the Voronoi diagram of $S$
within $R$. Hence, the list $L$ is complete.
\qed
\end{proof}

\begin{lemma}\label{lem: active}
For $i=0,1,\dots,O(\log n),$
the number of active basic squares in the grid $G_i(U)$
is $O(n)$ during the performance of  $DT-SQUARE(S,U).$
\end{lemma}
\begin{proof}
We argue similarly as at the beginning of the proof of Lemma \ref{lem: correct}.
If $R$ is an active basic square in $G_i(U)$
different from the unit
square $U$ then there must exist a basic square $Q$ in $TL_{i-1}(R')$,
where $R'$ is the parent of $R$
in $G_{i-1}(U),$
such that $Q$ contains
at least $100n/25$ points in $S.$
Now it is sufficient to note that: (i) there are at most
$O(n)$ basic squares in $G_{i-1}(U)$ that contain
at least $100n/25$ points in $S$; (ii) there are at most $O(1)$
basic squares $Q'$ in $G_{i-1}(U)$ different from $R'$ such that
$Q$ is included in $TL_{i-1}(Q')$; (iii) an active basic square in $G_{i-1}(U)$
can be a parent to at most four active basic squares in $G_i(U).$
\qed
\end{proof}

\begin{lemma}\label{lem: time}
The protocol $DT-SQUARE(S,U)$ can be implemented in $O(\log n)$
rounds on the congested clique.
\end{lemma}
\begin{proof}
Steps 1, 2 can be easily implemented in $O(1)$ rounds.
By Lemma \ref{lem: logn}, the block under the for loop in Step 3 is iterated
$O(\log n)$ times. It is sufficient to show that this block (a-g)
can be implemented in $O(1)$ rounds.

Step 3(a) can be performed totally locally.

The sorting of the prefixed representations of points in $S$ in Step 3(b)
can be done in $O(1)$ rounds by Fact 3.

For each node, the range of the numbers of the basic squares in $G_i(U)$ holding
the prefixed representations of points in $S$ at the node after the sorting of the prefix
representations of the points can be specified by two $O(\log n)$-bit numbers.
Hence, all nodes can inform all other nodes about their ranges in $O(1)$ rounds.
Thus, Step 3(b) requires $O(1)$ rounds in total.

The situation described in Step 3(c) can happen for at most $n$ basic squares
$W$ in $G_i(U).$ It requires sending by each node at most two different messages
to at most $n$ nodes in total and also receiving at most $n$ messages.
Hence, Step 3(c) can be implemented in $O(1)$ rounds by
using the routing protocol from Fact 2.

In Step 3(d), for each active basic square, a node representing the
square has to send $O(1)$ $O(\log n)$-bit queries to $O(1)$ other
nodes.  The total number of active basic squares in $G_i(U)$ is $O(n)$
by Lemma \ref{lem: active}.  Hence, by using the routing protocol from Fact 2
this task can be done in $O(1)$ rounds.

Consider Step 3(e). Answering the queries sent in Step 3(c) can be
done by local computations and the routing reverse to that in Step
3(d) in $O(1)$ rounds. After that each node for each active square
assigned to it determines locally if the criterion for computing the
Voronoi diagram of $S$ within $R$ is satisfied. If so the node sends
messages asking the nodes holding the prefixed representations of points
in the squares of $TL_i(R)$
for sending the
points. This requires sending $O(n)$ messages for each active basic
square in $G_i(U)$. Since the total number of such squares is
$O(n)$ by Lemma \ref{lem: active}  and each node represents
$O(1)$ active squares in  $G_i(U)$, it can be accomplished in
$O(1)$ rounds by Fact \ref{fact: routing}.
Delivering the requested points to the nodes
representing respective active basic squares can also be done in $O(1)$
rounds for the following reasons.  For each active basic square the node
representing it needs to receive $O(n)$ points. Furthermore, by Lemma
\ref{lem: active} there are $O(n)$ active basic squares in $G_i(U)$.
Hence, since the active squares are assigned to the $n$ nodes in a
balanced way, each node needs to receive $O(n)$ points. Also, the points
contained in a given basic square in $G_i(U)$ can be requested by at
most $O(1)$ nodes since there are at most $O(1)$ active basic squares
behind these requests to the given square.  Since the sorted prefixed representations of the
points in $S$ are divided between the nodes in a balanced way, each node
needs to send $O(n)$ points, each of them to $O(1)$ nodes. We conclude that this
part of Step 3(d) can be implemented in $O(1)$ rounds by
Fact~\ref{fact: routing}. The
remaining parts of Step 3(d) are done locally.

Step 3(f) requires sending and receiving by each node
$O(1)$ messages so it can be done in $O(1)$ rounds.

Consider an edge $(u,v)$ dual to some edge of the Voronoi diagram of the points of
$S$ included in $TL_i(R)$ within an active basic square $R$ in
$G_i(U).$  The edge can be appended to $L$ at most for $O(1)$ different squares
$R$ as $u,\ v$ are in $TL_i(R).$ Therefore, the list $L$ may contain
at most $O(1)$ copies of an edge of the Delaunay triangulation of $S$
so Step 3(g) can be implemented in $O(1)$ rounds by using the sorting
protocol
from Fact~\ref{fact: sorting}.
\qed
\end{proof}

Lemmata \ref{lem: correct}, \ref{lem: time} yield our first main result.

\begin{theorem}\label{theo: dlt}
    Let $S$ be a $(\frac 12,4\sqrt 2)$-smooth set of $n^2$ points with $O(\log
    n)$-bit coordinates in an orthogonal unit square, held in $n$-point  batches
    at the $n$ nodes of the congested clique. The set of edges
  of the Delaunay triangulation of $S$ dual to the edges of the
  Voronoi diagram of $S$ within the unit square can be constructed in
  $O(\log n)$ rounds on the congested clique.
\end{theorem}

\begin{lemma}\label{lem: dtvd}
    Let $S$ be defined as in Theorem \ref{theo: dlt}. Suppose that a list
  $L$ of the edges of the Delaunay triangulation of $S$ dual to the
  edges of the Voronoi diagram of $S$ within the unit square is held
  in $O(n)$-edge batches at the $n$ nodes of the congested clique. The
  Voronoi diagram of $S$ within the unit square can be constructed in
  $O(1)$ rounds on the congested clique.
  \end{lemma}
\begin{proof}
  Double the list $L$ by inserting for each $(u,v)\in L$ also
  $(v,u)$ into $L.$
  For each edge $(u,v)$ determine locally
  an $O(\log n)$-bit representation of the angle
  $\beta(u,v)$
  between $(u,v)$ and the horizontal line passing through $u$.
  For instance, the representation can specify the tangent
  of the angle by 
  $(v_y-u_y,  v_x-u_x).$
  Sort the edges $(x,y)$ by $(x,\beta(u,v))$,
  letting the nodes to compare the angle tangents locally,
    using the sorting protocol from Fact 3.
  In this way, for each point $u\in S,$  a sub-list of all edges of the Delaunay
  triangulation incident to $u$  in the angular order is created. Some of the sub-lists can stretch through
  several nodes of the clique network.
  Given the edges of the Delaunay triangulation incident to $u$
  in the angular order, the edges of the Voronoi region of $u$ within the unit square can be easily
  produced. This is done by intersecting the bisectors of $u$ and the other endpoints of consecutive
  edges incident to $u$ in the angular order as long as the intersection of two consecutive
  bisectors is within the unit square. Otherwise, the border of the region of $u$
  has to be filled with the fragment of the perimeter of the unit square between the
  intersections of the two bisectors with the perimeter.
  \qed
\end{proof}

Theorem \ref{theo: dlt} combined with Lemma \ref{lem: dtvd} yield
our second main result.

\begin{theorem}\label{theo: vor}
  Let $S$ be a $(\frac 12,4\sqrt 2)$-smooth set of $n^2$ points with $O(\log
  n)$-bit coordinates in an orthogonal unit square held in $n$-point
  batches at the $n$ nodes of the congested clique. The Voronoi
  diagram of $S$ can within the unit square be constructed in $O(\log
  n)$ rounds on the congested clique.
\end{theorem}

      \section{Final remarks}
      The {\em message complexity} of a protocol in the congested
      clique model is the maximum total number of $O(\log n)$-bit
      messages exchanged by the $n$ nodes of the congested clique
      during a run of the protocol (e.g., see \cite{PeS}).
      In case
      of our protocols, it is easily seen to be the product of the
      maximum number of messages that can be exchanged in a single
      round, i.e., $\Theta(n^2),$ times the number of required rounds.
      Thus, the message complexity of our deterministic protocols for
      the Delaunay triangulation and the Voronoi diagram of $n^2$
      point sets from Section 3 is $O(n^2\log n)$.
           
The remaining major  open problem is the derivation of a low
  polylogarithmic upper bound on the number of rounds sufficient
  to deterministically construct the Voronoi diagram of $n^2$ points
  with $O(\log n)$-bit coordinates in the Euclidean plane (when the
  points are not necessarily randomly distributed) on the congested
  clique with $n$ nodes.  This seems feasible but it might require a
  substantial effort as in the PRAM case \cite{AP95,ACG85}.

   Note here that the existence of an $O(\log n)$-time (unit cost)
  PRAM algorithm for a geometric problem on a point set (e.g., \cite{AP95}) does not
  guarantee the membership of the problem in the $NC^1$ class defined
  in terms of Boolean circuits \cite{FW23,NS22}. Simply, assuming that the input
points have $O(\log n)$-bit coordinates, the arithmetic operations of
the PRAM implemented by Boolean circuits of bounded fan-in have a
non-constant depth, at least $\Omega( \log \log n).$ Consequently, the
Boolean circuit simulating the $O(\log n)$-time PRAM for fixed input size can have
a super-logarithmic depth. This is a subtle and important point in the context of
relatively recent results of Frei and Wada  providing simulations of the
classes $NC^{k}$, $k>0,$  by MapReduce (see Theorems 9 and 10  in \cite{FW23}), and
consequently, in the Massively Parallel Computation (MPC) and BSP models
(see Theorem 1 in \cite{NS22}).
Due to the $O(1)$-round routing protocol of Lenzen \cite{L},
the congested clique model in our setting  can be roughly regarded as
a special case of MPC, where the size of the input is approximately
the square of the number of processors. For this reason, the $NC$ simulation
results from \cite{FW23} 
are relevant to our model only  when the parameter
$\epsilon$ in the exponent of the space bounds in \cite{FW23} equals $\frac 12$. This is possible in case
of Theorem 9 in \cite{FW23} on $NC^1$ simulation
but not possible in case of Theorem 10 in \cite{FW23} on $NC^k$,
k > 0, simulation. However, the proof of the former theorem
in \cite{FW23}  relies on
a strict logarithmic upper bound on the depth
of Boolean circuit of bounded fan-in 
required by Barrington's characterization of the $NC^1$ class
in terms of bounded-width polynomial-size branching programs.
Otherwise, one has to adhere to  the direct circuit simulation method from \cite{FW23}
that does not work for $\epsilon=\frac 12.$
In summary, Theorems 9 and 10 in \cite{FW23} do not  seem to have any direct consequences
for geometric problems on point sets in our model setting.
\junk{Reif and Sen \cite{RS} by
refining Clarkson's sequential random sampling method \cite{C88} for
the more general problem of the construction of the intersection of a
finite set of half-spaces in the three dimensional space could avoid
the merging bottleneck.  In effect, they obtained a work optimal,
randomized $O(\log n)$-time CREW PRAM algorithm for the Voronoi
diagram.  More efficient PRAM algorithms for the Voronoi diagram are
known for random planar point sets \cite{JLLP24,LKL,VV}.}
  \small
  \bibliographystyle{abbrv}
  \bibliography{Voronoi}
  \vfill
  \end{document}